\documentclass[10pt,a4paper]{article}
\usepackage{charter, amsmath, amsthm, amssymb, graphicx, subfigure, color, enumerate}
\usepackage[lmargin=30mm,rmargin=30mm,tmargin=25mm,bmargin=25mm]{geometry}
\renewcommand{\baselinestretch}{1.10}
\newcommand{\OPT}{OPT}
\newcommand{\ds}{\displaystyle}
\newtheorem{theorem}{Theorem}
\newtheorem{lemma}{Lemma}
\newtheorem{claim}{Claim}
\newtheorem{corollary}{Corollary}

\title{Minimum Entropy Combinatorial Optimization Problems\footnote{
A preliminary version of the work appeared in \cite{CiE}.
}}
\author{
Jean Cardinal\thanks{Universit\'e Libre de Bruxelles, D\'epartement d'Informatique, c.p.~212, B-1050 Brussels, Belgium, jcardin@ulb.ac.be.}
 \and 
Samuel Fiorini\thanks{Universit\'e Libre de Bruxelles, D\'epartement de Math\'ematique, c.p. 216,  B-1050 Brussels, Belgium,  sfiorini@ulb.ac.be.}
\and 
Gwena\"el Joret\thanks{Universit\'e Libre de Bruxelles, D\'epartement d'Informatique, c.p.~212,  B-1050 Brussels, Belgium, gjoret@ulb.ac.be. G. Joret is a Postdoctoral Researcher of the Fonds National de la Recherche Scientifique (F.R.S. -- FNRS).}
}

\date{}

\begin{document}
\maketitle
\sloppy

\begin{abstract}
We survey recent results on combinatorial optimization problems in which the objective function is the entropy of a discrete distribution.
These include the minimum entropy set cover, minimum entropy orientation, and minimum entropy coloring problems. 
\end{abstract}

\section{Introduction}

Set covering and graph coloring problems are undoubtedly among the most fundamental discrete optimization problems, and countless variants of these have been studied in the last 30 years. We discuss several coloring and covering problems in which the objective function is a quantity of information expressed in bits. More precisely, the objective function is the Shannon entropy of a discrete probability distribution defined by the solution. These problems are motivated by applications as diverse as computational biology, data compression, and sorting algorithms.\medskip

Recall that the entropy of a discrete random variable $X$ with probability distribution $\{p_i\}$, where $p_i := P[X=i]$, is defined as:
$$
H(X)= - \sum_i p_i \log p_i.
$$
Here, logarithms are in base 2, thus entropies are measured in bits; and $0\log 0 := 0$. From Shannon's theorem~\cite{Shanno48}, the entropy is the minimum average number of bits needed to transmit a random variable on an error-free communication channel.\medskip

\begin{figure}[t]
\begin{center}
\subfigure[\label{fig:exsc}set cover]{\includegraphics[scale=.33]{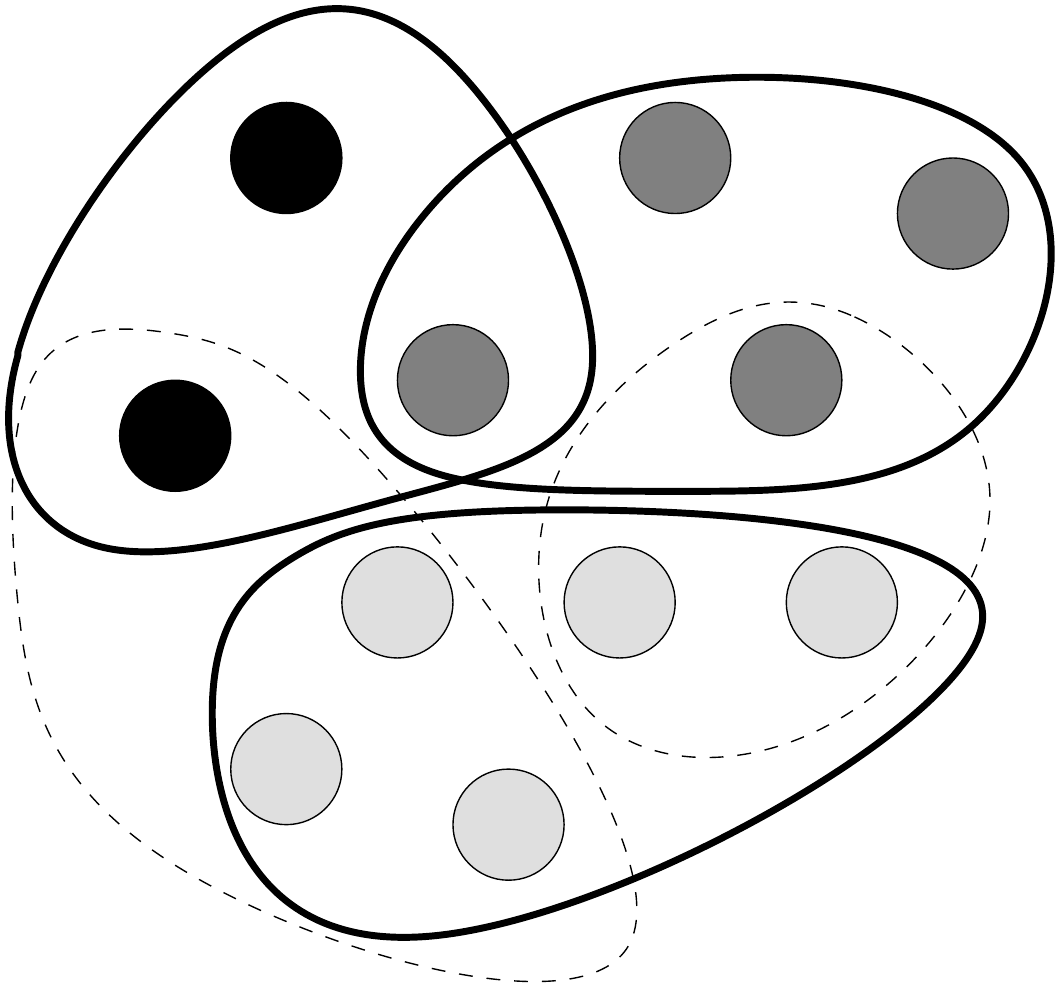}} \hspace{1cm} 
\subfigure[\label{fig:exor}orientation]{\includegraphics[scale=.4]{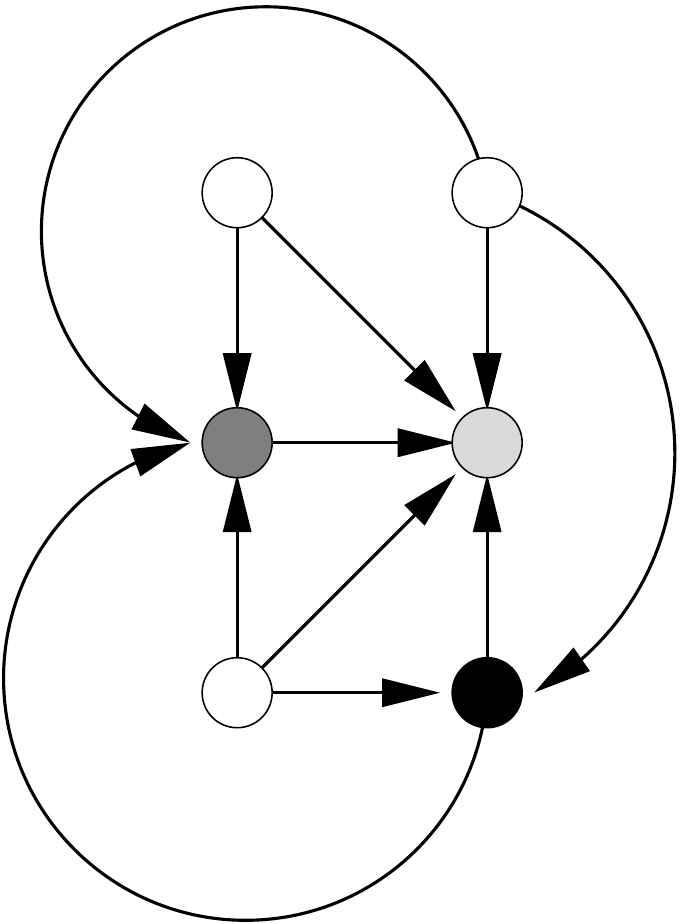}} \hspace{1cm}
\subfigure[\label{fig:excol}coloring]{\includegraphics[scale=.4, angle=90]{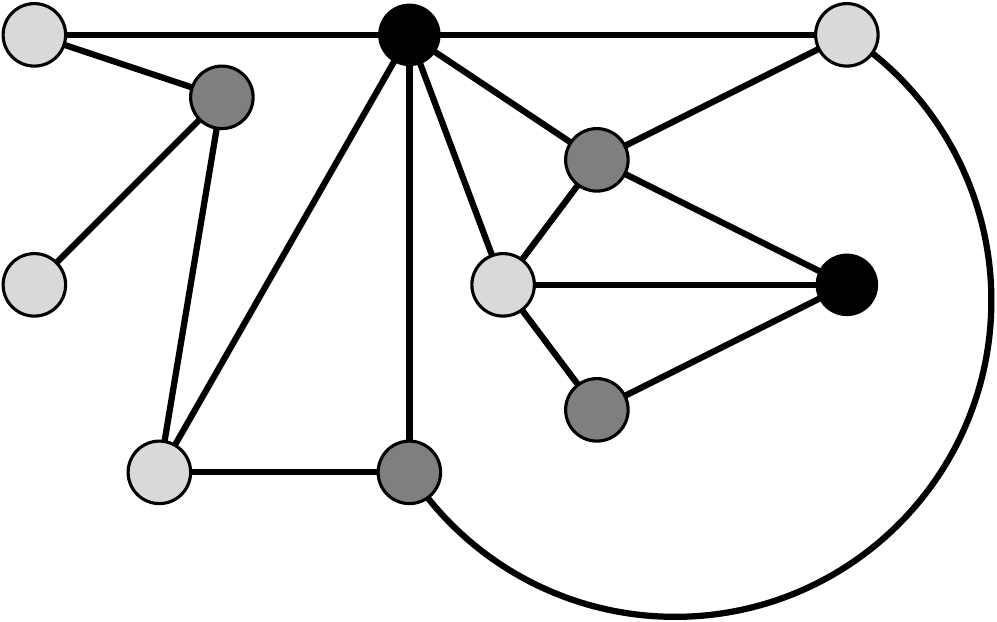}} \hspace{1cm}
\end{center}
\caption{\label{fig:ex}Instances of the minimum entropy combinatorial optimization problems studied in this paper, together
with feasible solutions. The resulting probability distribution for the given solutions is $\{ 5/11, 4/11, 2/11 \}$.}
\end{figure}

We give an overview of the recent hardness and approximability results 
obtained by the authors on three minimum entropy combinatorial optimization problems~\cite{CFJ07,CFJ08b,CFJ08a}. We also provide new approximability results on the minimum entropy orientation and minimum entropy coloring problems (Theorems~\ref{thm:ctapx}, \ref{hardchord}, \ref{thm:greedcol}, and \ref{thm:interv1}). Finally, we present a recent result that quantifies how well the entropy of a perfect graph is approximated by its chromatic entropy, with an application to a sorting problem~\cite{POP_SICOMP}.

\section{Minimum Entropy Set Cover}

In the well-known minimum set cover problem, we are given a ground set $U$ and a collection $\mathcal{S}$ of subsets of $U$, and we ask what is the minimum number of subsets from $\mathcal{S}$ such that their union is $U$. A famous heuristic for this problem is the greedy algorithm: iteratively choose the subset covering the largest number of remaining elements. The greedy algorithm is known to approximate the minimum set cover problem within a $1 + \ln n$ factor, where $n := |U|$. It is also known that this is essentially the best approximation ratio achievable by a polynomial-time algorithm, unless NP has slightly super-polynomial time algorithms~\cite{F98}.\medskip

In the minimum entropy set cover problem, the cardinality measure is replaced by the entropy of a partition of $U$ compatible with a given covering. The function to be minimized is the quantity of information contained in the random variable that assigns to an element of $U$ chosen uniformly at random, the subset that covers it. This is illustrated in Figure~\ref{fig:exsc}.
A formal definition of the minimum entropy set cover problem is as follows~\cite{HK05,CFJ08a}.
\begin{description}
\item[\sc instance:] A ground set $U$ and a collection $\mathcal{S} = \{S_1,\ldots,S_k\}$ of subsets of $U$
\item[\sc solution:] An assignment $\phi : U \to \{1,\ldots,k\}$ such that $x \in S_{\phi(x)}$ for all $x \in U$
\item[\sc objective:] Minimize the entropy $- \sum_{i=1}^k p_i \log p_i$, where $p_i := |\phi^{-1}(i)| / |U|$
\end{description}
Intuitively, we seek a covering of $U$ yielding part sizes that are either large or small, but somehow as nonuniform as possible. Also, an arbitrarily small entropy can be reached using an arbitrarily large number of subsets, making the problem quite distinct from the minimum set cover problem.

\paragraph{Applications.}

The original paper from Halperin and Karp~\cite{HK05} was motivated by applications in computational biology, namely haplotype reconstruction. In an abstract setting, we are given a collection of objects (that is, the set $U$) and, for each object, a collection of classes to which it may belong (that is, the collection of sets $S_i$ containing the object). Assuming that the objects are selected at random from a larger population, we wish to assign to each object the most likely class it belongs to.

Consider an assignment $\phi$, and suppose $q_i$ is the probability that a random object actually belongs to class $i$. We aim at maximizing the product of the probabilities for the solution $\phi$:
$$
\prod_{x\in U} q_{\phi (x)} .
$$
Now note that if $\phi$ is an optimal solution (supposing we know the values $q_i$), then the value $p_i := |\phi^{-1}(i)| / |U|$ is a maximum likelihood estimator of the actual probability $q_i$.
Thus, since the probabilities $q_i$ are unknown, we may replace $q_i$ by its estimated value:
\begin{equation}
\prod_{x\in U} p_{\phi (x)} = \prod_{i=1}^k p_i^{ |\phi^{-1}(i)|}.
\end{equation}
Maximizing this function is equivalent to minimizing the entropy $- \sum_{i=1}^k p_i \log p_i$, leading to the minimum entropy set cover problem.
In the haplotype phasing problem, objects and classes are genotypes and haplotypes, repectively. In a simplified model, a genotype can be modeled as a string in the alphabet $\{ 0, 1, ?\}$, and a haplotype as a binary string (see Figure~\ref{fig:haplotyping}). A haplotype {\em explains} or is {\em compatible} with a genotype if it matches it on every non-? position, and we aim at finding the maximum likelihood assignment of genotypes to compatible haplotypes. It is therefore a special case of the minimum entropy set cover problem. Experimental results derived from this work were proposed by Bonizzoni {\em et al.}~\cite{BVDM05}, and Gusev, M\u{a}ndoiu, and Pa\c{s}aniuc~\cite{GMIP08}.

\begin{figure}
\begin{center}
\includegraphics[scale=.5]{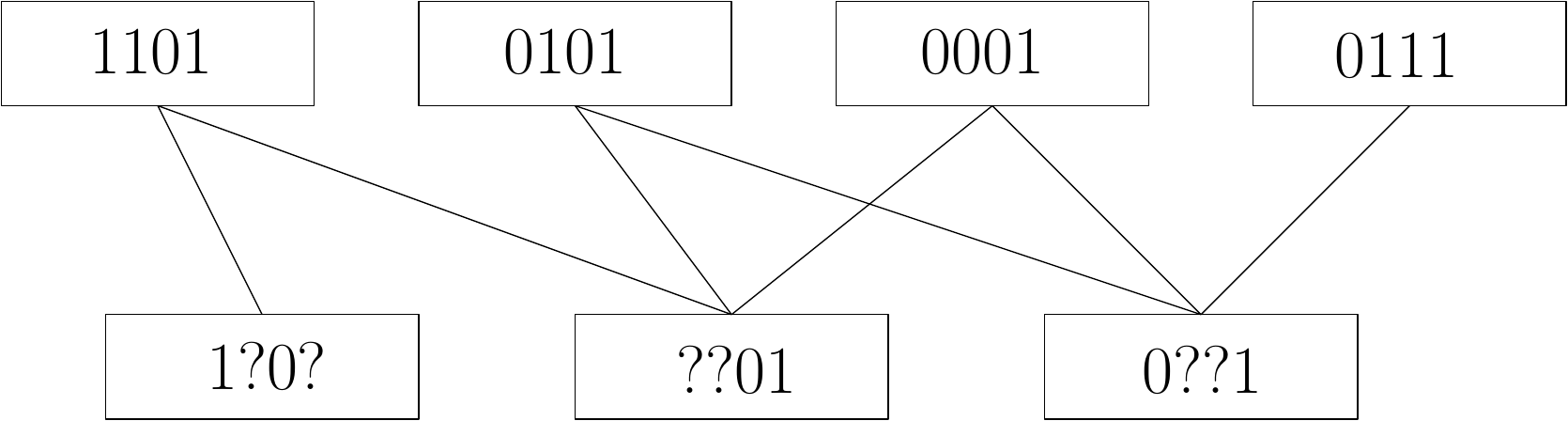}
\end{center}
\caption{\label{fig:haplotyping}The haplotype phasing problem (simplified setting).}
\end{figure}

\paragraph{Results.}

We proved that the greedy algorithm performs much better for the minimum entropy version of the set cover problem. Furthermore, we gave a complete characterization of the approximability of the problem under the $P\not= NP$ hypothesis.

\begin{theorem}[\cite{CFJ08a}]
\label{thm:apxsc}
The greedy algorithm approximates the minimum entropy set cover problem within an additive error of $\log e$ ($\approx 1.4427$) bits. Moreover, for every $\epsilon > 0$, it is NP-hard to approximate the problem within an additive error of $\log e - \epsilon$ bits.
\end{theorem}

Our analysis of the greedy algorithm for the minimum entropy set cover problem is an improvement, both in terms of simplicity and approximation error, over the first analysis given by Halperin and Karp~\cite{HK05}. The hardness of approximation is shown by adapting a proof from Feige, Lov\'asz, and Tetali on the minimum sum set cover problem~\cite{FLT04}, which itself derives from the results of Feige on minimum set cover~\cite{F98}.

\paragraph{Analysis of the Greedy Algorithm by Dual Fitting.}

We present another proof that 
the greedy algorithm approximates the minimum entropy set cover 
problem within an additive error of $\log e$. 
It differs from the proof given in~\cite{CFJ08a} 
in that it uses the dual fitting method.

Let $OPT$ denote the entropy of an optimum solution.
Let $\mathcal{S}^{*}$ be the collection of subsets of sets in $\mathcal{S}$, that is,
$\mathcal{S}^{*} := \{S: \exists S' \in \mathcal{S}, S \subseteq S'\}$. 
The following linear program gives a lower bound on $OPT$:
\begin{equation}
\label{primal}
\begin{array}{rl@{\qquad}l}
\min          &\ds \sum_{S \in \mathcal{S}^{*}} 
\left(-\frac{|S|}{n}\log \frac{|S|}{n}\right) \cdot x_{S} \\[2ex]
\textrm{s.t.} &\ds \sum_{S \in \mathcal{S}^{*}, S \ni v} x_{S} = 1 & \forall v\in U \\[2ex]
&\ds x_{S} \geq 0 & \forall S \in \mathcal{S}^{*}
\end{array}
\end{equation}
Indeed, if we add the requirements that $x_{S} \in \{0, 1\}$ for every set $S$, 
we obtain a valid integer programming formulation of the minimum entropy set cover problem.

The dual of \eqref{primal} reads:
\begin{equation}
\label{dual}
\begin{array}{rl@{\qquad}l}
\max          &\ds \sum_{v \in U} y_{v} \\[2ex]
\textrm{s.t.} &\ds \sum_{v \in S} y_{v} \leq -\frac{|S|}{n}\log \frac{|S|}{n} 
& \forall S \in \mathcal{S}^{*}
\end{array}
\end{equation}
%
In the case of the minimum entropy set cover problem,
the greedy algorithm iteratively selects a set $S\in \mathcal{S}$
that covers a maximum number of uncovered elements in $U$, and assigns the latter
elements to the set $S$. Let $\ell$ be the number of iterations performed by this algorithm 
on input $(U, \mathcal{S})$. For
$i\in \{1, \dots, \ell\}$, let $S_{i}$ be the set of elements that are
covered during the $i$th iteration. (Thus $S_{i} \in \mathcal{S}^{*}$.)

The entropy of the greedy solution is
\begin{equation}
\label{eq:greedy}
\sum_{i=1}^{\ell} -\frac{|S_{i}|}{n}\log \frac{|S_{i}|}{n} =: g.
\end{equation}

Now, for every $v\in U$, let $\tilde y_{v}$ be defined as
$$
\tilde y_{v} := - \frac{1}{n} \log \frac{|S_{i}| \cdot e}{n},
$$
where $i$ is the (unique) index such that $v \in S_{i}$.

If the vector $\tilde y$ is feasible for \eqref{dual}, then we deduce that
\begin{align*}
OPT &\geq \sum_{v \in U} \tilde y_{v}  \\
& = \sum_{i=1}^{\ell} -\frac{|S_{i}|}{n}\log \frac{|S_{i}| \cdot e}{n} \\
& = g - \log e,
\end{align*}
implying that the greedy algorithm approximates $OPT$ within an additive constant of $\log e$.
Hence, it is enough to prove that $\tilde y$ is feasible for the dual, that is, $\sum_{v \in S} \tilde y_{v} \leq -\frac{|S|}{n}\log \frac{|S|}{n}$ for every $S \in \mathcal{S}^{*}$.

Let $S \in \mathcal{S}^{*}$ and, for every  
$i\in \{1, \dots, \ell\}$, let $a_{i}:= |S \cap S_{i}|$. 
At the beginning of the $i$th iteration of the greedy algorithm, all the elements in 
$S - \cup_{j=1}^{i-1}S_{j}$ are not yet covered. 
Since the algorithm could cover all these elements at that iteration, we have
$|S_{i}| \geq |S - \cup_{j=1}^{i-1}S_{j}| = |S| - \sum_{j=1}^{i-1}a_{j}$. This implies
\begin{equation}
\label{eq:S}
\prod_{i=1}^{\ell} |S_{i}|^{a_{i}} \geq 
\prod_{i=1}^{\ell}\left(|S| - \sum_{j=1}^{i-1}a_{j}\right)^{a_{i}}.
\end{equation}

Using that $\sum_{i=1}^{\ell}a_{i} = |S|$, the following inequality is easily seen to hold:
\begin{equation}
\label{eq:factorial}
\prod_{i=1}^{\ell}\left(|S| - \sum_{j=1}^{i-1}a_{j}\right)^{a_{i}} \geq |S|!
\end{equation}
Combining \eqref{eq:S} and \eqref{eq:factorial}, and using the lower bound $|S|! \geq (|S| / e)^{|S|}$, 
we obtain
\begin{align*}
\sum_{v \in S} \tilde y_{v} &= \sum_{i=1}^{\ell} -\frac{a_{i}}{n}\log \frac{|S_{i}|\cdot e}{n} \\
&= -\frac{1}{n} \sum_{i=1}^{\ell} a_{i}\log\left( |S_{i}|\cdot e\right) + \frac{|S|}{n}\log n \\
&= -\frac{1}{n} \log \prod_{i=1}^{\ell} \left( |S_{i}|\cdot e\right)^{a_{i}} + \frac{|S|}{n}\log n \\
&= -\frac{1}{n} \log \left(e^{|S|} \prod_{i=1}^{\ell} |S_{i}|^{a_{i}} \right) + \frac{|S|}{n}\log n \\
&\leq -\frac{1}{n} \log \left(e^{|S|} |S|! \right) + \frac{|S|}{n}\log n \\
&\leq -\frac{1}{n} \log |S|^{|S|}+ \frac{|S|}{n}\log n \\
&= -\frac{|S|}{n}\log \frac{|S|}{n}.
\end{align*}
Therefore, $\tilde y$ is feasible for \eqref{dual}, as desired.

\section{Minimum Entropy Orientations (Vertex Cover)}

The minimum entropy orientation problem is the following~\cite{CFJ08b}:
\begin{description}
\item[\sc instance:] An undirected graph $G=(V,E)$ 
\item[\sc solution:] An orientation of $G$ 
\item[\sc objective:] Minimize the entropy $- \sum_{v\in V} p_v \log p_v$, where $p_v := \rho(v) / |E|$, and $\rho(v)$ is the indegree of
vertex $v$ in the orientation 
\end{description}

An instance of the minimum entropy orientation problem together with a feasible solution are given in Figure~\ref{fig:exor}. Note that the problem is a special case of minimum entropy set cover in which every element in the ground set $U$ is contained in exactly two subsets. Thus it can be seen as a minimum entropy vertex cover problem. 

\paragraph{Results.}

We proved that the minimum entropy orientation problem is NP-hard~\cite{CFJ08b}. Let us denote by $\OPT(G)$ the minimum entropy of an orientation of $G$ (in bits). An orientation of $G$ is said to be biased if each edge $vw$ with $\deg(v) > \deg(w)$ is oriented towards $v$. Biased orientations have an entropy that is provably closer to the minimum than those obtained via the greedy algorithm.
\begin{theorem}[\cite{CFJ08b}]
\label{th:biased}
The entropy of any biased orientation of $G$ is a most $\OPT(G) + 1$ bits. It follows that the minimum entropy orientation problem can be approximated within an additive error of $1$ bit, in linear time.
\end{theorem}

\paragraph{Constant-time Approximation Algorithm for Bounded Degree Graphs.}

By making use of the fact that the computation of a biased orientation is purely local \cite{local}, we show that we can randomly sample such an approximate solution to guess $\OPT(G)$ within an additive error of $1 + \epsilon$ bits. The complexity of the resulting algorithm does not (directly) depend on $n$, but only on $\Delta$ and $1 / \epsilon$. This is a straightforward application of ideas presented by Parnas and Ron~\cite{PR07}. 

We consider a graph $G$ with $n$ vertices, $m$ edges, and maximum degree $\Delta$. Pick any preferred biased orientation $\overrightarrow{G}$ of $G$. (For instance, we may order the vertices of $G$ arbitrarily and orient an edge $vw$ with $\deg(v) = \deg(w)$ towards the vertex that appears last in the ordering.) The following algorithm returns an approximation of $\OPT(G)$ (below, $s$ is a parameter whose value will be decided later).

{\tt
\begin{enumerate}
\item For $i = 1$ to $s$
\begin{enumerate}
\item pick vertex $v_i$ uniformly at random
\item compute the indegree $\rho (v_i)$ of $v_i$ in $\overrightarrow{G}$
\end{enumerate}
\item return $H := \log m - \frac {n}{sm} \sum_{i=1}^s \rho (v_i) \log \rho (v_i)$
\end{enumerate}
}

The worst-case complexity of the algorithm is $O(s\Delta^2)$. 

\begin{theorem}
\label{thm:ctapx}
There is an algorithm of worst-case complexity $O(\Delta^4 \log^2\Delta / \epsilon^2)$ that, when given a graph $G$ with maximum degree $\Delta$ and at least as many edges as vertices, returns a number $H$ satisfying, with high probability,
$$
\OPT(G) \leq H \leq \OPT(G) + (1 + \epsilon).
$$
\end{theorem}
\begin{proof}
Let $V := V(G)$ and $\OPT := \OPT(G)$. From the previous theorem we get:
$$
-\sum_{v\in V} \frac{\rho (v)}{m} \log \frac{\rho (v)}{m} = \log m - \frac 1m \sum_{v\in V} \rho (v) \log \rho (v) \leq \OPT + 1.
$$
For $i = 1, \ldots, s$, let $v_i$ denote a uniformly sampled vertex of $G$. By linearity of expectation, we have:
$$
E \left[\sum_{i=1}^s \rho(v_i) \log \rho(v_i)\right] = \frac sn \sum_{v\in V} \rho (v) \log \rho (v).
$$
Noting $0 \le \rho(v_i) \log \rho(v_i) \le \Delta \log \Delta$ (for all $i$), Hoeffding's inequality then implies:
\begin{eqnarray}
P\left[ \left| \sum_{i=1}^s \rho (v) \log \rho (v) - E\left[ \sum_{i=1}^s \rho (v) \log \rho (v)\right] \right| \geq \epsilon s \right] &  \leq  &
2\exp \left( - \frac{2s^2\epsilon^2}{s(\Delta \log \Delta)^2} \right) \nonumber \\
\Rightarrow
P\left[ \left| \frac n{sm} \sum_{i=1}^s \rho (v) \log \rho (v) - \frac 1m \sum_{v\in V} \rho (v) \log \rho (v) \right| \geq \epsilon \frac nm \right] &  \leq  &
2\exp \left( - \frac{2s\epsilon^2}{(\Delta \log \Delta)^2} \right)\nonumber \\
\Rightarrow
P\left[ \left| H - \OPT \right| \geq 1 + \epsilon \right] \leq
P\left[ \left| H - \OPT \right| \geq 1 + \epsilon \frac nm \right] & \leq & 2\exp \left( - \frac{2s\epsilon^2}{(\Delta \log \Delta)^2} \right). \nonumber
\end{eqnarray}
By letting $s=\Theta((\Delta \log \Delta)^2 / \epsilon^2)$, with arbitrarily high probability, we conclude that the above algorithm provides an approximation of $\OPT$ within $1+\epsilon$ bits in time $O(s\Delta^2) = O(\Delta^4\log^2\Delta / \epsilon^2)$. Note that this approximation can be either an under- or an over-approximation. To make the approximation one-sided, we can simply return $H + \epsilon$.
\end{proof}

\section{Minimum Entropy Coloring}

A proper coloring of a graph assigns colors to vertices such that adjacent vertices have distinct colors. We define the entropy of a proper coloring as the entropy of the color of a random vertex. An example is given in Figure~\ref{fig:excol}. The minimum entropy coloring problem is thus defined as follows:
\begin{description}
\item[\sc instance:] An undirected graph $G=(V,E)$ 
\item[\sc solution:] A proper coloring $\phi : V\to \mathbb{N}^{+}$ of $G$ 
\item[\sc objective]: Minimize the entropy $- \sum_i p_i \log p_i$, where $p_i := |\phi^{-1}(i)| / |V|$
\end{description}

Note that any instance of the minimum entropy coloring problem can be seen as an implicit instance of the minimum entropy set cover problem, in which the ground set is the set of vertices of the graph, and the subsets are all independent sets, described implicitly by the graph structure.\medskip

The problem studied by the authors in~\cite{CFJ07} was actually slightly more general: the graph $G$ came with nonnegative weights $w(v)$ on the vertices $v \in V$, summing up to $1$. The weighted version of the minimum entropy coloring problem is defined similarly as the unweighted version except now we let $p_i := \sum_{v\in\phi^{-1}(i)} w(v)$. (The unweighted version is obtained for the uniform weights $w(v) = 1/|V|$.)

\paragraph{Applications.}

\begin{figure}
\begin{center}
\subfigure[\label{fig:sideinformation}]{\includegraphics[scale=.5]{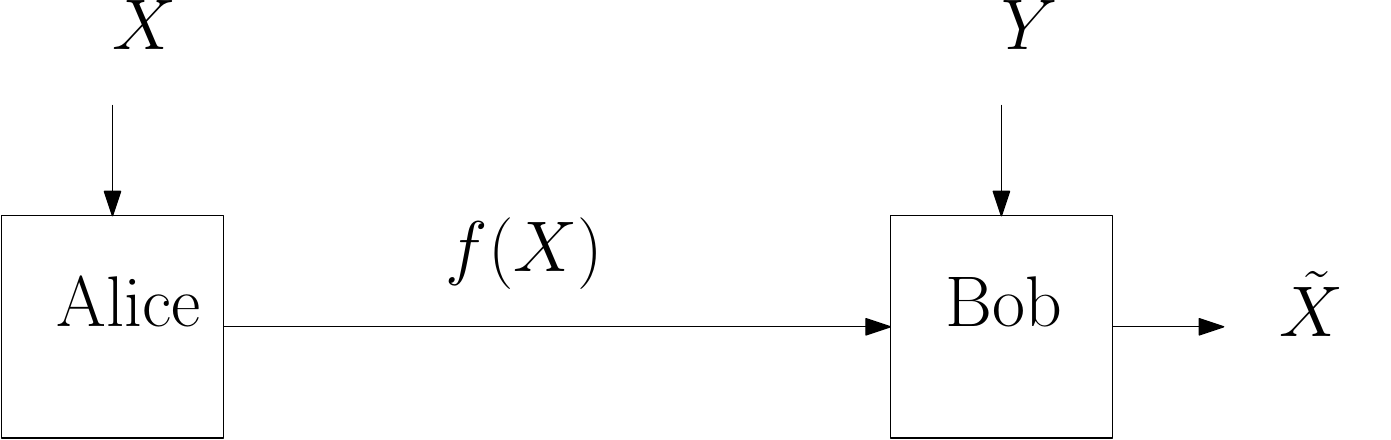}}
\ \\
\subfigure[\label{fig:jointprob}Joint probability distribution]{\mbox{}\qquad
\begin{tabular}{|c|c|c|}
\hline
       & 0 & 1 \\
       \hline
$a$ & $>0$ & $>0$ \\
\hline
$b$ & $>0$ & 0 \\
\hline
$c$ & 0 & $>0$ \\
\hline
\end{tabular}\qquad\mbox{}
}
\hspace{2cm}
\subfigure[\label{fig:confusability}Confusability graph.]{\mbox{}\qquad\includegraphics[scale=.4,angle=-90]{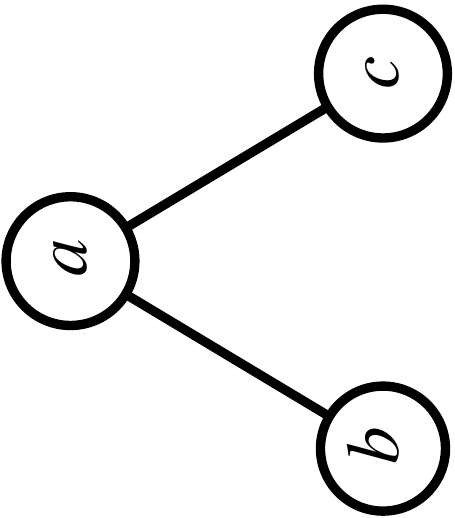}\qquad\mbox{}} 
\end{center}
\caption{Coding with side information.}
\end{figure}

Minimum entropy colorings have found applications in the field of data compression, and are related to several results in 
zero-error information theory~\cite{zeit}. The minimum entropy of a coloring is called the chromatic entropy by Alon and Orlitsky~\cite{AO96}. It was introduced in the context of coding with side information, a source coding scenario in which the receiver has access to an information that is correlated with the data being sent (see also \cite{KTRR03}). 

The scenario is pictured on Figure~\ref{fig:sideinformation}. Alice wishes to send a random variable $X$ to Bob, who has access to a side information $Y$. The value $Y$ is unknown to Alice, but allows Bob to gain some information on $X$ (thus $X$ and $Y$ are not independent). We need to find an efficient code that allows Alice and Bob to communicate. This code will consist of an assignment of binary codewords to the possible values of $X$.

To give a simple example, suppose that $X\in \{a,b,c\}$, and $Y\in \{0,1\}$. The joint probability distribution $P(X,Y)$ of $X$ and $Y$ is such that $P(b, 1)=0$, $P(c,0)=0$, and $P(x,y)>0$ for all other pairs $x,y$. We can notice that if Alice assigns the same codeword to the values $b$ and $c$, Bob will always be able to lift the ambiguity, and tell that $X=b$ if his side information $Y$ has value $0$, and $X=c$ if $Y=1$. More generally, from the joint probability distribution of $X$ and $Y$ (which is supposed to be known to both parties), we can infer a {\em confusability graph} with vertex set the domain of $X$ (here $\{a,b,c\}$) and an edge between two values $x,x'$ whenever there exists $y$ such that $P(x,y)>0$ and $P(x',y)>0$. Any coloring of the confusability graph will yield a suitable code for Alice and Bob. The rate of this code is exactly the entropy of the coloring, taking into account the marginal probability distribution $P(X)$. Note that the confusability graph does not depend on the exact values of the probability, but only on whether they are equal to 0 or not, which is typical of the zero-error setting~\cite{zeit}.

Minimum entropy colorings are instrumental in several other coding schemes introduced by Doshi {\em et al.}~\cite{DSMJ07,DSM07,DSMJ06} for functional data compression. It was also proposed for the encoding of segmented images~\cite{ABicip02}. Another application is described in Section~\ref{sec:ent}.

\paragraph{Results.}

Unsurprisingly, the minimum entropy coloring problem is hard to solve even on restricted instances, and hard to approximate in general. We proved the following two results~\cite{CFJ07}.

\begin{theorem}
\label{thm:hardinterval}
Finding a minimum entropy coloring of a weighted interval graph is strongly NP-hard.
\end{theorem}

The following hardness result is quite strong because the trivial coloring assigning a different color to each vertex has an entropy of at most $\log n$. Actually, the approximation status of the problem is very much comparable to that of the maximum independent set problem~\cite{H99}. (This is not a coincidence, see for instance the discussion below and in particular Corollary~\ref{cor:perfect_col}.)

\begin{theorem}
\label{thm:minentcol_inapx}
For any positive real $\epsilon$, it is NP-hard to approximate the minimum entropy coloring problem within an additive error of $(1 - \epsilon) \log n$.
\end{theorem}

A positive result was given by Gijswijt, Jost, and Queyranne:
\begin{theorem}[\cite{GJQ07}]
The minimum entropy coloring problem can be solved in polynomial time on weighted co-interval graphs.
\end{theorem}

\paragraph{Hardness for Unweighted Chordal Graphs.}

The reduction given in~\cite{CFJ07} to prove Theorem~\ref{thm:hardinterval} uses in a crucial
way the weights on the vertices. It is therefore natural to ask whether the problem is also
NP-hard on {\em unweighted} interval graphs. While we do not know the answer to this question,
we show here that, in its unweighted version, the minimum entropy coloring problem 
is NP-hard on chordal graphs (which contain interval graphs). Our proof is a variant
of the previous reduction; the main ingredient is a gadget that allows us to (roughly) 
simulate the weights.

\begin{theorem}
\label{hardchord}
The minimum entropy coloring problem is NP-hard on unweighted chordal graphs.
\end{theorem}

Before proving Theorem~\ref{hardchord}, 
we introduce a few definitions and lemmas. 
Suppose $q=(q_i)$ and $r=(r_i)$ are two probability distributions
over $\mathbb{N}^+$ with finite support. If $\sum_{i=1}^\ell q_i \le \sum_{i=1}^\ell r_i$
holds for all $\ell$, we say that $q$ is {\em dominated} by $r$.
The following lemma is a standard consequence of the strict
concavity of the function $x \mapsto -x\log x$
(see~\cite{FHN-soda, HLP88} for different proofs).

\begin{lemma}
\label{lem-jungle}
Let $q=(q_i)$ and $r=(r_i)$ be two probability distributions 
over $\mathbb{N}^+$ with finite support. Assume that $q$ is 
nonincreasing, that is, $q_i \ge q_{i+1}$ for $i\ge 1$. If 
$q$ is dominated by $r$, then the entropy of $q$ is at least that of $r$, with
equality if and only if $q=r$.
\end{lemma}

A coloring $\phi$ of a graph $G$ {\em realizes} a probability distribution
$q$ if $q$ is the probability distribution induced by $\phi$, that is, 
if $q_{i} = |\phi^{-1}(i)|/n$ for all $i$ (where $n:=|G|$).

Let $J_k$ be the intersection 
graph of the set $\{h_i^j\ |\ 1 \le j \le i \le k\}$ of intervals, 
where $h_i^j$ denotes the open interval $\big((j-1)/i, j/i\big)$ 
(see Figure~\ref{fig-J5} for a representation of $J_5$).
We call the independent set $\{h_i^j\ |\ 1 \le j \le i\}$ 
the {\em $i$-th row} of $J_k$. 

\begin{figure}
\centering
\includegraphics[width=0.5\textwidth]{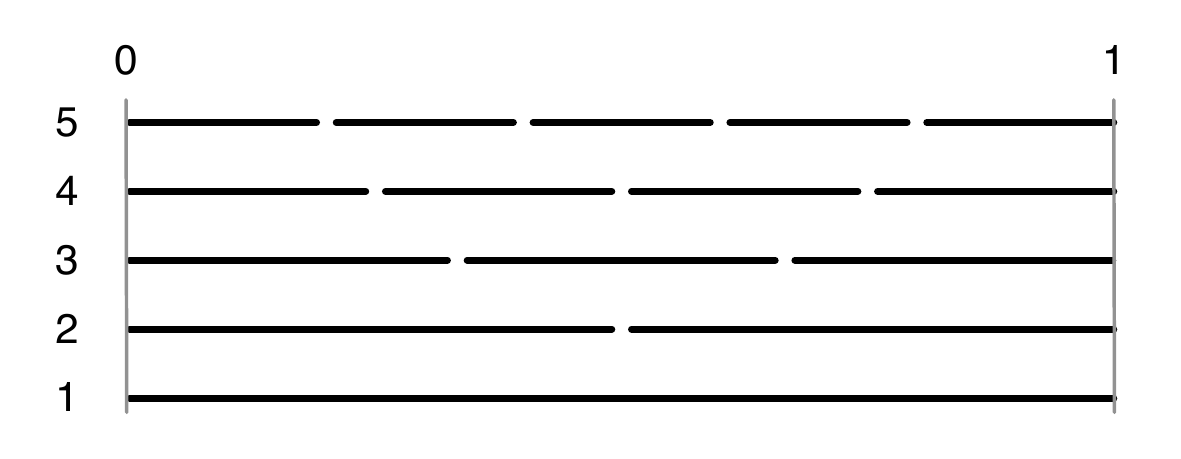} 
\caption{\label{fig-J5}An interval representation of $J_{5}$.}
\end{figure}

\begin{lemma}
\label{lem-Jk}
The probability distribution $\big(k/|J_{k}|, (k-1)/|J_{k}|, \dots, 1/|J_{k}|, 0, \dots\big)$
dominates every probability distribution realizable by a coloring 
of $J_{k}$. Moreover, every coloring 
realizing this distribution colors each row with a single color.
\end{lemma}
\begin{proof}
For $v\in V(J_{k})$, denote by $|v|$ the length of the interval corresponding
to $v$. Fix $i\in\{1,\dots,k\}$ and consider any subpartition $P$ of 
$V(J_{k})$ into $i$ nonempty parts $P_{1}, \dots, P_{i}$ satisfying
\begin{equation}
\label{eq-lem-Jk-1}
\sum_{v\in P_{j}} |v| \leq 1
\end{equation}
for every $1\leq j\leq i$. By the definition of $J_{k}$ it follows 
\begin{equation}
\label{eq-lem-Jk-2}
|P_{1}| + \cdots + |P_{i}| \leq k + (k-1) + \cdots + (k-i+1),
\end{equation}
with equality if and only if the parts of $P$ correspond to
the rows numbered from $k-i+1$ to $k$.

Now, consider any coloring $\phi$ of $J_{k}$ and denote
by $P_{1}, \dots, P_{\ell}$ its color classes.
Since each color class of $\phi$ is an independent set, $P_{j}$ must 
satisfy~\eqref{eq-lem-Jk-1} for every $j$. The lemma follows then
by using~\eqref{eq-lem-Jk-2} for every $i\in \{1,\dots, \ell\}$.
\end{proof}

\begin{proof}[Proof of Theorem~\ref{hardchord}] 
We reduce from the NP-complete problem of deciding if a circular arc
graph $G$ is $k$-colorable~\cite{GJMP80}.
Given a circular arc graph $G$, there exists a polynomial time algorithm to construct a circular representation
of it~\cite{T80}. A key idea of our proof is to start with
a circular arc graph and cut it open somewhere to obtain an interval
graph (see~\cite{M05-orl} for another application of this technique to minimum sum coloring).

Let $y$ be an arbitrary point on the circle that is not the endpoint
of any arc in the representation of $G$. Denote by $k'$
the number of arcs in which $y$ is included. If $k'>k$, then $G$ is not
$k$-colorable. On the other hand, if $k'<k$, we add to the representation $k-k'$ sufficiently
small arcs that only intersect arcs including $y$. This cannot
increase the chromatic number of $G$ beyond $k$. Thus we assume 
that $y$ is contained in exactly $k$ arcs.

We denote by $a_1,\ldots, a_k$ the arcs containing $y$. Splitting each
arc $a_i$ into two parts $\ell_i$ and $r_i$ at point $y$ yields an interval
representation of an interval graph $G'$. The original graph $G$ is
$k$-colorable if and only if there exists a $k$-coloring of $G'$ in which
$\ell_j$ and $r_j$ receive the same color for $1 \leq j \leq k$.

Up to this point, the reduction is the same as for weighted interval graphs~\cite{CFJ07}.
The latter proceeds by adding weights on the vertices of the interval graph $G'$.
Here, we will instead subdivide each interval $\ell_i$ and $r_i$.
We first describe the transformation for the intervals $r_i$. 
Consider an interval representation of $G'$ where each interval $r_i$ is of the form 
$r_i = (y_i,-1) \cup [-1,1)$ for some real $y_i < -1$. We split each 
$r_i$ into $i+1$ intervals $r^0_i, r_i^1, 
\ldots ,r_i^i$ with 
$$
r_i^j := 
\begin{cases}
\big(y_i,-i/k \big) &\text{ if } j=0;\\
\big(-i/k, 1/i \big)  & \text{ if } j=1,\\
\big((j-1)/i, j/i \big) & \text{ otherwise}.\\   
\end{cases}
$$
(See Figure~\ref{fig-J5-extended} for an illustration.)
Notice that the subgraph induced by $R := \{r_i^j\ |\ 
1 \le j \le i \le k\}$ is isomorphic to $J_k$. 
A symmetric modification is made on the intervals $\ell_i$, 
the new intervals are denoted $\ell^0_i$, $\ell_i^1$ ,\ldots, $\ell_i^i$. 
We also let $L := \{\ell_i^j\ |\ 1 \le j \le i \le k\}$, and denote by $J$ the graph 
obtained after the transformation on the intervals $r_{i}$ and $\ell_{i}$ is completed.

\begin{figure}
\centering
\includegraphics[width=0.9\textwidth]{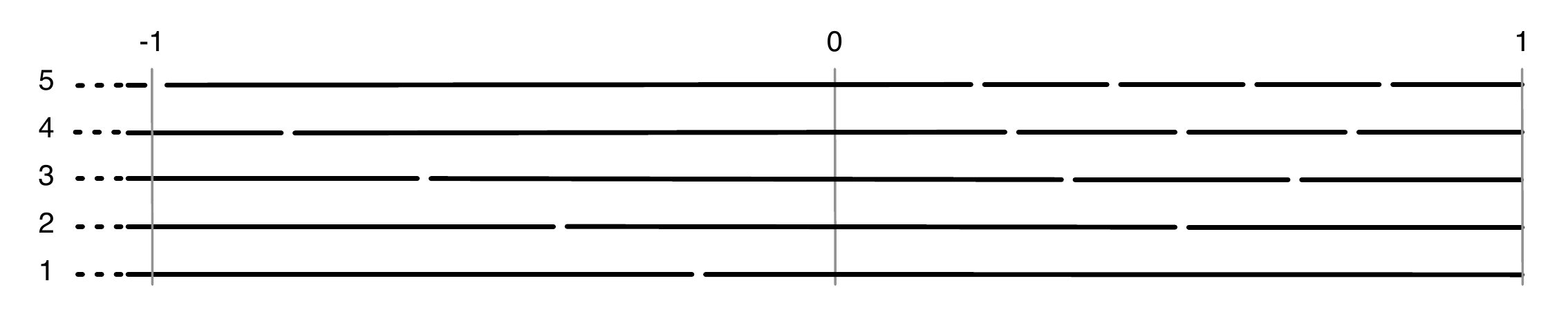} 
\caption{\label{fig-J5-extended}Splitting of the $r_{i}$'s for $k=5$.}
\end{figure}

Now, we hang on each vertex $v\in V(J) - (L\cup R)$  
a clique of cardinality $k$. Thus $v$ is one of the vertices of the clique 
and the $k-1$ other vertices are new vertices that are added to $J$.
The resulting graph $J'$ is not necessarily an interval graph, but it 
is still chordal.

Let $q^*=(q^*_{i})$ be the probability distribution over $\mathbb{N}^{+}$ defined as
$$
q^*_{i}:= \left\{
\begin{array}{lll}
\big(n + 2(k-i+1)\big) / |J'| & & \textrm{if } i\in \{1,\dots,k\}; \\
0 & & \textrm{otherwise},
\end{array}
\right.
$$
where $n := |G'|$.

\begin{claim}
\label{claim-hardness-chordal-dominates}
$q^*$ dominates every probability distribution realizable by a coloring of $J'$.
\end{claim}
\begin{proof}
Let $\phi$ be any coloring of $J'$ and denote by $q=(q_{i})$ the probability 
distribution realized by $\phi$. 
Let $U:=V(J') - (L\cup R)$  and write $q$ as $q=q^{R} + q^{L} + 
q^{U}$, where $q^{R}$ ($q^{L},q^{U}$) denotes the sequence $q$ where
only the contribution of vertices in $R$ ($L$, $U$ respectively) is taken into account.

By Lemma~\ref{lem-Jk}, the sequence $q^{R}$ is dominated by 
$$
q^{*, J_{k}} := \big(k/ |J'|, (k-1)/ |J'|, \dots, 1/ |J'|, 0, \dots\big).
$$
(Although the dominance relation was originally defined for probability distributions only, it 
naturally extends to a pair $q,r$ of sequences of
nonnegative real numbers with $\sum_{i\geq 1} q_{i} = \sum_{i\geq 1} r_{i}$.) 
The same holds for $q^{L}$. Moreover, $q^{U}$ is dominated by
$q^{*, U}$, where 
$$
q^{*, U}_{i}:= \left\{
\begin{array}{lll}
n / |J'| & & \textrm{if } i\in \{1,\dots,k\}, \\
0 & & \textrm{otherwise}.
\end{array}
\right.
$$
This is easily seen using the fact that $V(J')$ can be partitioned into $n$ cliques
of cardinality $k$. 
Since $q^{*} =  2q^{*, J_{k}} + q^{*, U}$, we deduce that $q^{*}$
dominates $q$. Hence, Claim~\ref{claim-hardness-chordal-dominates} holds.
\end{proof}

\begin{claim}
\label{claim-hardness-chordal-iff}
$G$ is $k$-colorable if and only if there exists a coloring of $J'$ realizing $q^*$.
\end{claim}
\begin{proof}
By the definition of $J'$, a $k$-coloring of $G$ can be readily extended
to a $k$-coloring of $J'$ realizing $q^*$. In order to prove the other
direction of the claim, assume that $\phi$ is a coloring of $J'$ realizing $q^{*}$.
Using the second part of Lemma~\ref{lem-Jk}, we know that $\phi$
colors the vertices of $R$ rowwise, that is, it assigns color $k-i+1$
to all vertices of the $i$-row $\{r_i^j\ |\ 1 \le j \le i\}$.  Also, as $\phi$ uses exactly
$k$ colors and $r_i^0$ is adjacent to  $r_{i'}^1$ for $1\leq i < i' \leq k$,
the vertex $r_i^0$ is also colored with color $k-i+1$. Since the same observations
apply also on $L$, from $\phi$ we easily derive a $k$-coloring of $G$.
Claim~\ref{claim-hardness-chordal-iff} follows.
\end{proof}

By combining Claims~\ref{claim-hardness-chordal-dominates} and~\ref{claim-hardness-chordal-iff}
with Lemma~\ref{lem-jungle}, we deduce the following: If
$G$ is $k$-colorable, then every minimum entropy coloring of $J'$ realizes $q^{*}$.
Furthermore, if $G$ is not $k$-colorable, then
no coloring of $J'$ realizes $q^{*}$. Therefore, by computing a minimum entropy coloring 
of the chordal graph $J'$, we could decide if the circular arc graph $G$ is $k$-colorable.
Theorem~\ref{hardchord} follows.
\end{proof}

We remark that, as pointed out by a referee,
the graph $J'$ in the above reduction is actually {\em strongly chordal}
(every cycle of even length at least $6$ has a chord splitting the cycle in two even cycles).
Thus the minimum entropy coloring problem remains NP-hard on (unweighted) strongly chordal graphs.
We note that interval graphs are a proper subclass of strongly chordal graphs, which in turn are a proper subclass of chordal graphs (see~\cite{BVS99}).

\paragraph{Greedy Coloring.}

The greedy algorithm can be used to find approximate minimum entropy colorings. In the context of coloring problems, the greedy algorithm involves iteratively removing a maximum independent set in the graph, assigning a new color to each removed set. This algorithm is polynomial for families of graphs in which a maximum (size or weight) independent set can be found in polynomial time, such as perfect graphs. 
We therefore have the following result.

\begin{corollary}
\label{cor:perfect_col}
The minimum entropy coloring problem can be approximated within an additive error of $\log e\ (\approx 1.4427)$ bits when restricted to perfect graphs.
\end{corollary}

We now show an example of an ``approximate greedy" algorithm yielding a bounded approximation error on any bounded-degree graphs. 

\begin{lemma}
\label{lem:beta-appx}
The algorithm that iteratively removes a $\beta$-approximate maximum independent set yields an approximation of the minimum entropy of a coloring within an additive error of $\log \beta + \log e$ bits.
\end{lemma}
\begin{proof}
The entropy of a coloring $\phi$ can be rewritten as (letting $n:=|V|$):
\begin{equation}
\label{logtoprod}
- \sum_i \frac{|\phi^{-1}(i)|}{n}  \log \frac{|\phi^{-1}(i)|}{n} = -\frac 1n \sum_{v\in V} \log \frac{|\phi^{-1}(\phi(v))|}{n} .
\end{equation}
Thus minimizing the entropy of a coloring $\phi : V\to \mathbb{N}^{+}$ of $G$ is equivalent to maximizing the following product of the color class sizes over all vertices:
$$
\prod_{\phi} := \prod_{v\in V} |\phi^{-1}(\phi (v))| .
$$
Consider a color class $S$ in an optimal coloring $\phi_{OPT}$ and the order in which the vertices of this subset are colored by the approximate greedy algorithm (breaking ties arbitrarily).
Let us denote by $\phi_G$ the coloring obtained with the approximate greedy algorithm.
The first vertex in this order is assigned by $\phi_G$ to a color class that has size at least $|S|/\beta$ (since at that stage, we know there exists an independent set of size at least $|S|$). 
The next vertex is assigned to a color class of size at least $(|S|-1)/\beta$, and so on.
Hence the product of the color class sizes restricted to the elements of $S$ is at least $|S|!/\beta^{|S|}$. 
In the optimal solution, however, this product is $|S|^{|S|}$ by definition. 
Hence, denoting by $S_i$ the color classes of an optimal solution, the approximate greedy algorithm yields at least the following product:
\begin{eqnarray}
\prod_{\phi_G}\geq \prod_i \frac{|S_i|!}{\beta^{|S_i|}} & \geq & \prod_i \frac{|S_i|^{|S_i|}}{(e\cdot \beta)^{|S_i|}} = \frac{\prod_{\phi_{OPT}}}{(e\cdot \beta)^n},
\end{eqnarray}
Converting this back to entropies using equation~\ref{logtoprod} yields the claimed result.
\end{proof}

\begin{theorem}
\label{thm:greedcol}
Minimum entropy coloring is approximable in polynomial time within an additive error of $\log (\Delta + 2) - 0.1423$ bits on graphs with maximum degree $\Delta$.
\end{theorem}
\begin{proof}
It is known that for graphs with maximum degree at most $\Delta$, a $\rho$-approximate maximum 
independent set can be computed in polynomial time for $\rho := (\Delta  + 2)/3$
(see Halld\'orsson and Radhakrishnan~\cite{greed}).
We have
$$
\log \rho + \log e = \log (\Delta  + 2) + \log e - \log 3 < \log(\Delta  + 2) - 0.1423.
$$
Hence, the claim follows from Lemma~\ref{lem:beta-appx}.
\end{proof}

\paragraph{Coloring Interval Graphs.}

We give the following simple polynomial-time algorithm for approximating the minimum entropy coloring problem on
(unweighted) interval graphs. This algorithm is essentially the online coloring algorithm proposed by Kierstead and Trotter~\cite{KT},
but in which the intervals are given in a specific order, and the intervals in $S_i$ are 2-colored offline.

{\tt
\begin{enumerate}
\item sort the intervals in increasing order of their right endpoints; let $(v_1, v_2,\ldots ,v_n)$ be this ordering
\item {\bf for} $j\gets 1$ to $n$ {\bf do}
\begin{itemize}
\item insert $v_j$ in $S_i$, where $i$ is the smallest index such that 
$S_1\cup S_2\cup\ldots S_i\cup\{ v_j\}$ does not contain an $(i+1)$-clique
\end{itemize}
\item color the intervals in $S_1$ with color 1
\item let $k$ be the number of nonempty sets $S_i$
\item {\bf for} $i\gets 2$ to $k$ {\bf do}
\begin{itemize}
\item \label{2col} color the intervals in $S_i$ with colors $2i-2$ and $2i-1$
\end{itemize}
\end{enumerate}
}

This algorithm is also similar to the {\tt BETTER-MCA} algorithm introduced by Pemmaraju, Raman, and Varadarajan for the max-coloring 
problem~\cite{PRK04}, but instead of sorting the intervals in order of their weights, we sort them in order of their right endpoints. It achieves the same 
goal as the algorithm of Nicoloso, Sarrafzadeh, and Song~\cite{NSS99} for minimum sum coloring.

\begin{lemma}
\label{mics}
At the end of the algorithm, the graph induced by the vertices in $S_1\cup S_2 \cup \ldots \cup S_i$ is a maximum $i$-colorable subgraph of $G$.
\end{lemma}
\begin{proof}
If we consider the construction of the set $S_1\cup S_2\ldots \cup S_i$, 
it matches exactly with an execution of the algorithm of Yannakakis and Gavril~\cite{YG87} for constructing a maximum $i$-colorable subgraph.
\end{proof}

Kierstead and Trotter~\cite{KT} proved the following lemma, that shows that step~\ref{2col} of the algorithm is always feasible. 
\begin{lemma}[\cite{KT}]
\label{2colSi}
The graphs induced by the sets $S_i$ are bipartite.
\end{lemma}

Let $H'$ be the entropy of the probability distribution $\{ |S_i| / n\}$. The proof of the following relies on Lemma~\ref{mics}. 
\begin{lemma}
\label{lem:lb}
$H'$ is a lower bound on the minimum entropy of a coloring of the interval graph.
\end{lemma}
\begin{proof}
Let us consider a coloring with color classes $Q_i$.
We can assume without loss of generality that the classes $Q_i$ are labeled in order of decreasing sizes, that is, $|Q_{i+1}|\leq |Q_i|$. 
For any $t\leq k$:
$$
\sum_{i=1}^t |S_i| \geq \sum_{i=1}^t |Q_i|,
$$
since from Lemma~\ref{mics} the set $S_1\cup S_2\cup\ldots \cup S_t$ induces a maximum $t$-colorable subgraph of $G$.
It follows from Lemma~\ref{lem-jungle} 
that the entropy of the distribution $\{|S_i|/n\}$ is smaller or equal to that of $\{|Q_i|/n\}$.
\end{proof}

We deduce the following.
\begin{theorem}
\label{thm:interv1}
On (unweighted) interval graphs, the minimum entropy coloring problem can be approximated within an additive error of 1 bit, in polynomial time.
\end{theorem}
\begin{proof}
The entropy of the coloring produced by the algorithm is at most that of the distribution $\{ |S_i| / n\}$ plus one bit, since
the intervals in $S_i$ are colored with at most two colors. From Lemma~\ref{lem:lb} the former is a lower bound on the optimum, 
and we get the desired approximation.
\end{proof}

To conclude this section, we mention that
improved approximation results could be obtained on other special classes of graphs, 
using for instance the methods developed by Fukunaga, Halld\'orsson, and Nagamochi~\cite{FHN-soda,FHN07}.

\section{Graph Entropy and Partial Order Production}
\label{sec:ent}

The notion of graph entropy was introduced by K\"orner in 1973~\cite{K73}, and was initially motivated by a source coding problem.
It has since found a wide range of applications (see for instance the survey by Simonyi~\cite{S95}). 

The entropy of a graph $G=(V,E)$ can be defined in several ways. We now give a purely combinatorial (and not information-theoretic) definition, and restrict ourselves to unweighted (or, more precisely, uniformly weighted) graphs. Let us consider a probability distribution $\{ q_S\}$ on the independent sets $S$ of $G$, and denote by $p_v$ the probability that $v$ belongs to an independent set drawn at random from this distribution (that is, $p_v := \sum_{S \ni v} q_S$). Then the entropy of $G$ is the minimum over all possible distributions $\{ q_S\}$ of
\begin{equation}
\label{eq:minlogp}
- \frac1n \sum_{v\in V} \log p_v.
\end{equation}
The feasible vectors $(p_v)_{v\in V}$ form the stable set polytope $STAB(G)$ of the graph $G$, defined as the following convex combination:
$$
STAB (G) := \mathrm{conv} \{ \mathbf{1}^S : S \text{ independent set of } G \},
$$
where $\mathbf{1}^S$ is the characteristic vector of the set $S\subseteq V$, assigning the value 1 to
vertices in $S$, and 0 to the others. Thus, the entropy of $G$ can be written as
$$
H(G) := \min_{p\in STAB(G)} -\frac1n \sum_{v\in V} \log p_v.
$$

The relation between the graph entropy and the chromatic entropy (the minimum entropy of a coloring) can be made clear
from the following observation: if we restrict the vector $p$ in the definition of $H(G)$ to be a convex combination of characteristic vectors of {\em disjoint\/} independent sets, then the minimum of (\ref{eq:minlogp}) is equal to the chromatic entropy (see~\cite{CFJ07} for a rigorous development). Hence, the mathematical program defining the graph entropy is a relaxation of that defining the chromatic entropy. It follows that the graph entropy is a lower bound on the chromatic entropy.\medskip

Recently, we showed that the greedy coloring algorithm, that iteratively colors and removes a maximum independent set, yields a good approximation of the graph entropy, provided the graph is perfect.

\begin{theorem}[\cite{POP_SICOMP}]
\label{thm:colapxent}
Let $G$ be a perfect graph, and $g$ be the entropy of a greedy coloring of $G$. Then
$$
g \leq H(G) + \log (H(G) + 1) + O(1).
$$
\end{theorem}

The proof of Theorem \ref{thm:colapxent} is a dual fitting argument based on the identity $H(G) + H(\bar{G}) = \log n$, that holds whenever $G$ is perfect (Czisar, K\"orner, Lov\'asz, Marton and Simonyi~\cite{CKLMS90} characterized the perfect graphs as the graphs that ``split entropy'', for all probability distribution on the vertices). In particular, Theorem~\ref{thm:colapxent} implies that the chromatic entropy of a perfect graph never exceeds its entropy by more than $\log \log n + O(1)$. It turns out that this is essentially tight~\cite{POP_SICOMP}.\medskip

Theorem~\ref{thm:colapxent} is the key tool in a recent algorithm for the partial order production problem~\cite{POP_SICOMP}. In this problem, we want to bijectively map a set $T$ of objects coming with an unknown total order to a vector equipped with partial order on its positions, such that the relations of the partial order are satisfied by the mapped elements. The problem admits the selection, multiple selection, sorting, and heap construction problems as special cases. Until recently, it was not known whether there existed a polynomial-time algorithm performing this task using a near-optimal number of comparisons between elements of $T$. By applying (twice) the greedy coloring algorithm and the approximation result of Theorem~\ref{thm:colapxent}, we could reduce, in polynomial time, the problem to a well studied multiple selection problem and solve it with a near-optimal number of comparisons. The reader is referred to \cite{POP_SICOMP} for further details.

\paragraph{Acknowledgment.}
We thank the two anonymous referees for their helpful comments on a previous version of the manuscript.
We are especially grateful to one referee for pointing out an error in a previous
version of the analysis of the greedy algorithm by dual fitting.
This work was supported by the Communaut\'e Fran\c caise de Belgique (projet ARC).

\bibliographystyle{plain}
\bibliography{CiE.bib}

\end{document}